\documentclass[amsmath, floatfix,nofootinbib, prx, aps, usenames, dvipsnames,10pt]{quantumarticle}

\usepackage[utf8]{inputenc}
\usepackage[english]{babel}
\usepackage[T1]{fontenc}

\usepackage[usenames,dvipsnames]{xcolor}
\usepackage{graphicx}
\usepackage[ colorlinks = true, 
             linkcolor = MidnightBlue,
             urlcolor  = MidnightBlue,
             citecolor = orange,
             anchorcolor = ForestGreen,
]{hyperref}

\usepackage[sans]{dsfont} 
\usepackage{bbm}
\usepackage[mathscr]{euscript}
\usepackage{cancel} 

\usepackage{tcolorbox}
\tcbuselibrary{theorems} 

\usepackage{amsmath}
\usepackage{amssymb}
\usepackage{amsthm}
\usepackage{units}
\usepackage{framed}
\usepackage{mdframed}
\usepackage{thmtools, thm-restate} 
\usepackage{ucs}
\usepackage{subcaption}
\usepackage{fullpage}

\makeatletter
\def\l@subsubsection#1#2{}
\def\l@subsection#1#2{}
\makeatother

\tcbset{colback=orange!5,colframe=orange!75!yellow, 
fonttitle= \bfseries, float=htb}

\newtcbtheorem[]{bigexample}{Example}{float*=htb, width=\textwidth}{ex}

\newtcolorbox[blend into=figures]{boxfigure}[2][]
{ float*=htb,width=\textwidth,lower separated=false, center upper, 
center title,title={#2},every float=\centering,#1}

\declaretheoremstyle[shaded={rulecolor=MidnightBlue,rulewidth=1pt, bgcolor={rgb}{1,1,1}}]{boxed}

\declaretheoremstyle[shaded={rulecolor=Thistle,rulewidth=1pt, bgcolor={rgb}{1,1,1}}]{secboxed}

\declaretheoremstyle[shaded={rulecolor=YellowOrange,rulewidth=1pt, bgcolor={rgb}{1,1,1}}]{terboxed}

\declaretheoremstyle[shaded={rulecolor=Green,rulewidth=1pt, bgcolor={rgb}{1,1,1}}]{tetraboxed}

\declaretheorem[]{lemma}

\declaretheorem[sibling=lemma, style=tetraboxed]{corollary}

\declaretheorem[sibling=lemma, style=tetraboxed]{proposition}
\declaretheorem[style=secboxed]{definition}

	\newcommand{\blue}[1]{\textcolor{Blue}{#1}}
	\newcommand{\green}[1]{\textcolor{OliveGreen}{#1}}

	\newcommand{\orange}[1]{\textcolor{Orange}{#1}}
	
	\newcommand{\flag}[1]{\green{ [#1]}}

    \newcommand{\ket}[1]{\vert  #1 \rangle}

	\newcommand{\hilbert}{\mathcal{H}}

    \newcommand{\com}[1]{ \overline{#1}\, }

	\newcommand{\ball}{\mathcal{B}}

	\newcommand{\tr}{\operatorname{Tr}  }
	
\newcommand*{\E}{\mathcal{E}}

\newcommand{\cT}{\mathcal{T}}

\newcommand*{\h}{{\bf h}}

	\newcommand*{\eps}{\varepsilon}

\begin{document}

\title{Operational locality in global theories}

\author{Lea Kr\"amer}
\affiliation{Institute for Theoretical Physics, ETH Zurich, Switzerland}

\author{L\'idia del Rio}
\affiliation{Institute for Theoretical Physics, ETH Zurich, Switzerland}
\affiliation{School of Physics, University of Bristol, United Kingdom}
\email{delrio@phys.ethz.ch}

\date{}

\maketitle

\begin{abstract}
Within a global physical theory, a notion of locality allows us to find and justify information-processing primitives, like non-signalling between distant agents. 
Here we propose exploring the opposite direction: to take agents as the basic building blocks through which we test a physical theory, and recover operational notions of locality from signalling conditions. 
First we introduce an operational model for the  effective state spaces of individual agents, as well as the range of their actions.  We then formulate natural secrecy conditions between agents  and identify the aspects of locality relevant for signalling. We discuss the possibility of taking commutation of transformations as a primitive of physical theories, as well as applications to quantum theory and generalized probability frameworks.
This ``it from bit'' approach establishes an operational connection between local action and local observations, and gives a global interpretation to concepts like discarding a subsystem or composing local functions. 
\end{abstract}

\begin{figure*}[t]
    \centering
    
    \includegraphics[width=0.85\textwidth]{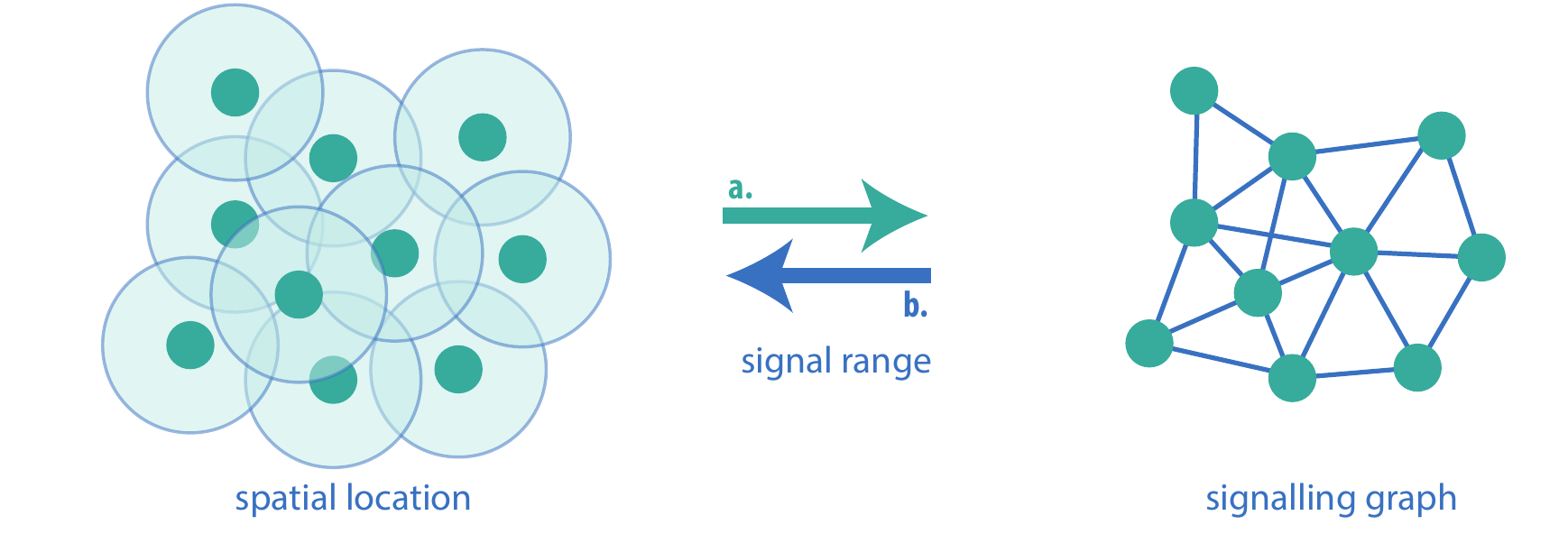}

    \caption{{\bf There and back again: physical locality and signalling.} On the left, the spatial location of several agents (dark dots) and their range of communication (overlapping circles) are depicted; on the right, the corresponding signalling graph.
    {\bf a.}  In physical theories, the notion of a space-time background where agents are positioned, together with principles about the range of signalling (e.g.\ the finite speed of light) allow us to derive information-processing concepts like non-signalling agents.  {\bf b.} Reverse direction: starting from the notion of agents that may or not be able to communicate and minimal assumptions on the nature and range of signalling, it may be possible to deduce both the space-time structure of the theory and the position of agents in it, to a good approximation.
    We can take inspiration from a simple example in the field of localization in wireless sensor networks  (for a review see e.g.\ Ref.~\cite{Cheng2012}). Forest fire prevention mechanisms can be implemented by dropping a large number of smoke-detecting sensors from a plane over the forest. The sensors (our agents) are equipped with short-range communication systems, and land at random positions. One then collects the data of which sensors can signal to each other. From the signalling graph, it is possible to reconstruct the relative positions of the sensors on the ground to high accuracy --- that way, when the smoke alarm goes off in a sensor, the fire-response team can quickly locate it.
}

    \label{fig:detectors}
\end{figure*}

In modelling local agents acting within a global theory, the intuitive assumption is that both their actions and their knowledge are restricted to a bounded region. The canonical example is a scientist who has full control of her lab and can perform local tomography.
In reality though, the breadth of knowledge and the range of action of agents may be decoupled. For example, prisoners can acquire global knowledge by reading the news, but their actions are limited to small subsystems. Conversely, someone locked in a control room may only have local knowledge of the shapes of different buttons, but  pressing one may have global consequences. 
The observation that the knowledge and action do not always go hand in hand implies that in order to model agents we have to specify both (Section~\ref{sec:local_agents}). 
This naturally leads us to search for minimal operational constraints needed to ensure that agents are truly local.

Here we  motivate a notion of \emph{secrecy} between agents, which captures whether actions performed by an agent (like writing a message, choosing a bit or preparing a quantum state) can be perceived by another (Section~\ref{sec:non-signalling}); traditional notions of non-signalling correspond to an extended secrecy between space-like separated regions  (Section~\ref{sec:discussion}). 
This work brings together and clarifies concepts of locality used in quantum theory, generalized probabilistic theories and field theories.  It highlights that the \emph{state space} and \emph{transformations} of a theory are but a  subjective choice of representation of the underlying physical theory from a  viewpoint that is convenient to a given agent, as argued by Spekkens~\cite{Spekkens2012}. Here, we tentatively suggest commutation of transformations as a primitive  of physical theories. In particular, we show how 
to derive local agents (and effective descriptions of local subsystems) from commutation relations on global transformations  (Section~\ref{sec:commutativity}). 

This work draws from our ``Resource theories of knowledge'' \cite{DelRio2015}, and has natural applications in multi-player settings, like cryptographic scenarios, games or resource theories. 
There is yet a more exciting possible application: to recover the space-time structure of a physical theory from the primitive notion of test agents, in the spirit of Hardy's operational GR~\cite{Hardy2016} and to the task of localization in wireless sensor networks~\cite{Cheng2012}. The idea is to send out agents (or probes) to unknown positions, see if they can communicate with each other, and use the signalling graph to define distances between agents,  reconstruct their relative positions, and infer properties of space-time
(Figure~\ref{fig:detectors}). 
For this we must first find appropriate, theory-independent notions of  agents and signalling.

\section{Modelling agents}
\label{sec:local_agents}

We start with a top-down approach, where we first describe a global theory (as seen by a global agent), and then model restricted agents acting within that theory.

\subsection{Global theory}

From the point of view of a given global agent, a global theory may be represented  via a state space $\Omega$ and a set of  transformations $\cT$ that are available to the agent
\cite{Coecke2014, Fritz2015,Coecke2014b,Coecke2011}.
We can think of the state space as the ``language'' chosen by this global observer to describe nature. For example, $\Omega$ could be the set of coordinates and momenta of all celestial bodies; in quantum theory, it could be the set of valid density matrices over a global Hilbert space. It need not be a static picture: in astronomy, an alternative state space $\Omega'$ could  be the set of possible trajectories of celestial bodies, and in quantum theory it could include all  global Hamiltonians that determine the free evolution of density matrices.
Note that
{\bf 1)} $\Omega$ is not the ultimate description of reality, just a convenient representation from the point of view of a global agent;
{\bf 2)} different pictures, like  $\Omega$ and  $\Omega'$, may be related and mapped to one another \cite{DelRio2015,Spekkens2012}; and
{\bf 3)} $\Omega$ need not  have any special structure \emph{a priori} besides being a set --- indeed, the approach laid out here will allow us to find an operational subsystem structure in the set of states.

The transformations in $\cT$ represent all actions that the theory allows the global agent to implement. 
We can think of them as the ways in which the agent may test a theory, by applying actions that change state parameters. 
For example, an explicit theory of a quantum universe may allow only for unitary operations, while a more generous theory could equip the agent with implicit large ancillas, and allow her to implement general quantum channels, state preparations and even tomography. Again the two views can be related: the latter is an \emph{effective theory} derived from the unitary quantum theory, by internalizing part of the global space as belonging to the agent and her instruments, and not to the object of study (the rest of the universe) \cite{DelRio2015}.  In the context of field theories, this is discussed as  \emph{emerging agency}  \cite{Hardy2016}. 
In a superdeterministic theory, there is only one possible course of evolution for the universe, and  $\cT$ consists only of functions that apply it (for example $\cT \cong \{e^{-i H t} \}_t$ where the global agent is given some choice of time).
Formally, $\cT$ is a monoid of functions $f: \Omega \to \Omega$: it contains the identity transformation and is closed under concatenation (an associative binary operation), such that performing two actions subsequently, $f \circ g$, is still an allowed operation. We discuss the monoidal assumption and possible relaxations in Section~\ref{sec:discussion}.

\subsection{Local agents}

\begin{figure*}[t]
    \centering
    
    \includegraphics[width=\textwidth]{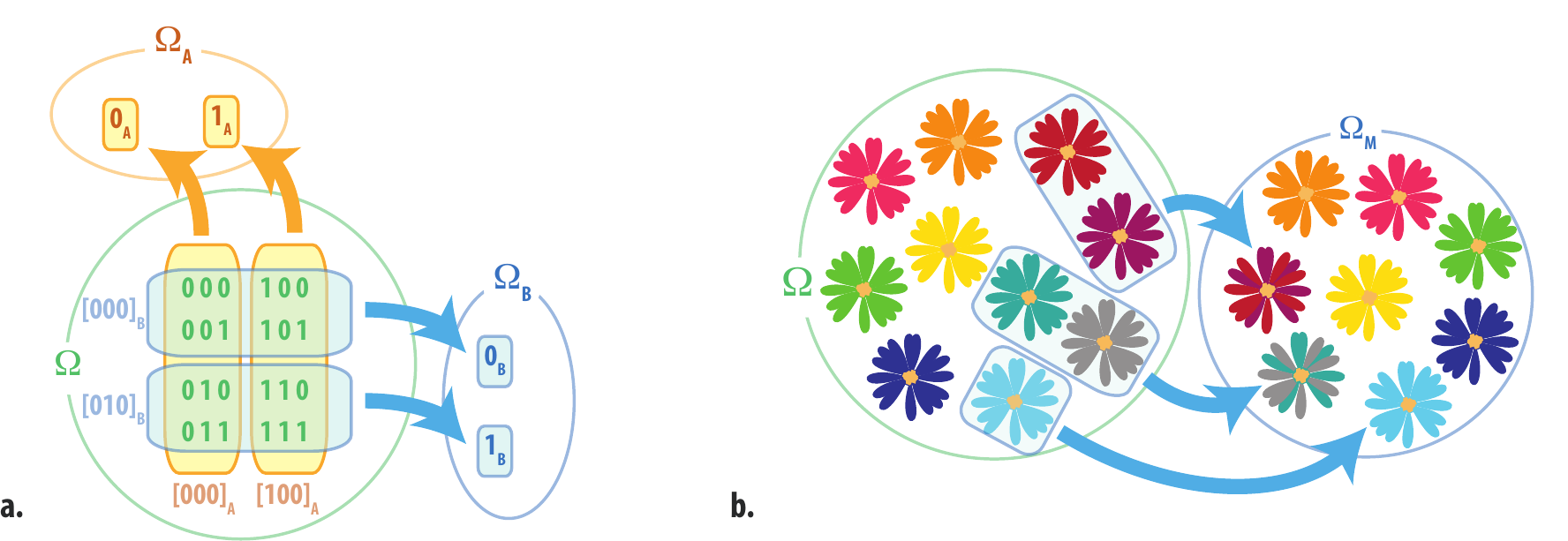}

    \caption{{\bf Building an agent's effective state space.} The different states of a global space $\Omega$ are shown to an agent, who\  finds equivalence classes of (subjectively) indistinguishable states. Their effective state space is then the quotient space. {\bf a.} The global state space $\Omega$ consists of three bits, in the eight possible states depicted. An agent Alice can only see the first bit, therefore she cannot distinguish the states in each vertical box (her equivalence classes $[000]_A$ and $[100]_A$). Her effective state space $\Omega_A = \Omega/\!\sim_A$ has only two states, which can be relabelled as $0_A$ and $1_A$ for convenience. Another agent Bob identifies the equivalent classes $[000]_B$ and $[010]_B$, which leads us to conclude that he can only see the second bit. Note that for example if Alice were able to apply transformations that only change the first bit, she could not signal to Bob (because he could not detect the change).
    {\bf b.} Here
    $\Omega$ is the space of colours, which were shown to a partly colourblind agent Marvin. Marvin identified the colours that he could not distinguish, which allowed us to build his reduced state space of colours $\Omega_M$. 
}

    \label{fig:quotient}
\end{figure*}


Local agents are characterized by limited knowledge: their inability to distinguish global states that appear identical in their eyes. We can formalize this by building equivalence classes of states that are indistinguishable from the perspective of an agent. 
For example, in quantum theory, we could have an agent Bob who only has access to a Hilbert space $\hilbert_B$; two global states are indistinguishable (or equivalent) from Bob's perspective if they have the same marginal in  $\hilbert_B$. This defines an equivalence relation $\sigma \sim_B \rho: \tr_{\com B} \sigma = \tr_{\com B} \rho $, where $\tr_{\com B}$ denotes a partial trace over all systems except $B$. 
The corresponding equivalence classes are
\begin{align*}
    [\rho]_B := \{\sigma \in \Omega: \tr_{\com B} \sigma =  \rho_B \}.
\end{align*}
Taking the quotient over this equivalence relation gives us a new space state $ \Omega /\! {\sim_B}  $, which is in one-to-one correspondence with  the set  of all reduced density matrices in $\hilbert_B$. This is Bob's \emph{effective state space}, sufficient to encode all the information that he can observe about any global state (Figure~\ref{fig:quotient}). 
In this case, the map from the global to the local spaces (the \emph{canonical map}) is given by the partial trace: 
\begin{align*}
    \h_B:  \Omega &\to \Omega /\! {\sim_B}  \\
    \rho &\mapsto  [\rho]_{B} \quad  \cong  \tr_{\com B} \rho = \rho_B.
\end{align*}
More generally, we can always build the effective state space of an agent in this way, even if we do not know anything about the structure of the global space (for instance whether it can be split into a convenient tensor form $\hilbert_A \otimes \hilbert_B$). The construction of an agent's effective state space $\Omega_B := \Omega /\! {\sim_B}$ is in the spirit of Leibniz principle of identity of indiscernibles  \cite{leibniz1739}. Yet this operational procedure emphasizes that both discernibility and identity are subjective concepts (Figure~\ref{fig:quotient}).
Limitations on Bob's perspective may have nothing to do with spatial locality. Bob might only have access to crude measurement instruments unable to distinguish microscopic details of states, or he may not be able to distinguish a global phase or gauge \cite{Hardy2016}. In generalized probability frameworks, Bob's perspective can correspond to a grouping of individual global outcomes into events  (Appendix~\ref{appendix:GPT}).  In algebraic quantum field theory, these equivalence classes could emerge from algebras of local observables (see e.g.\ Ref.~\cite{Valente2013} for a review).

The other ingredient needed to define an agent, as we saw in the introduction, is a description of the  actions available to him. As his actions may have a global impact, a minimal approach is to take them to be a submonoid $\cT_B\subseteq\cT$ of the globally allowed transformations. We discuss relaxations of this definition in Section~\ref{sec:discussion}. 
Generalizations of this approach can be found in Ref.~\cite{DelRio2015}. There, we also study explicit ways to move between global and local views (technically, related by \emph{Galois insertions}), effective theories and other properties of local agents. 

\begin{definition}[Global theory and restricted agents]
A \emph{global theory of agents} is defined by a pair 
 $(\Omega,\cT)$, where $\Omega$  (the \emph{state space}) is a set, and $\cT$ is a monoid of transformations $f: \Omega \to \Omega$, with the concatenation operation $\circ$. 

A \emph{restricted agent} $B$ acting within the theory is defined by a pair $(\sim_B, \cT_B)$ , where $\sim_B$ is an equivalence relation in $\Omega$ and $\cT_B$ is a submonoid of $\cT$ called the set of \emph{local operations} of the agent. The quotient space
$\Omega_B:= \Omega/\!\sim_B$ is called the \emph{effective space} of agent $B$. 
The \emph{reduction} to the effective space is given by the canonical map 
\begin{align*}
    \h_B: \Omega &\to \Omega_B \\
         \rho &\mapsto [\rho]_B .
\end{align*}
\end{definition}

We can always further coarse-grain the effective state space of a given agent $B$ in order to obtain a more restricted agent $C$. For example,  in renormalization group flow, lowering the cutoff corresponds to coarse-graining over more and more observables \cite{Wilson1974,Polchinski1984}. The following proposition formalizes this idea \cite[Prop.\ III.5]{DelRio2015}. All proofs can be found in Appendix~\ref{appendix:proofs}.

\begin{restatable}[Nested agents]{proposition}{PropNestedAgents}
\label{prop:nested_agents}
Let $(\Omega, \cT)$ be a global theory, and $B$, $C$ two restricted agents. Then the following are equivalent: 
\begin{enumerate}
    \item $C$ has more restricted knowledge than $B$, that is $[\rho]_B \subseteq [\rho]_{C}, \quad \forall \ \rho \in \Omega $, 
    \item There exists an equivalence relation $ \sim_{B \to C}$ in $B$'s effective state space $\Omega_B$ such that  $\Omega_C \cong \Omega_B / \sim_{B \to C}$.
\end{enumerate}

\end{restatable}

\section{Secrecy between agents}
\label{sec:non-signalling}

\subsection{Secrecy}

Having defined agents, we may study conditions for secrecy and non-signalling between them. 
Consider a setup of two agents Alice and Bob, represented by $A =(\sim_A,\cT_A)$ and $B =(\sim_B,\cT_B)$. 
Imagine that Alice wants to keep her actions (like writing a message or preparing a state) secret from Bob. This is achieved if  Bob cannot tell whether she applied them, even after post-processing.\footnote{Bob's effective space may  include his local processing  (``states that I can distinguish after applying all my accessible operations") or not (``states that I distinguish immediately, before further processing"). For the sake of generality, we leave the freedom in this decision up to the agent, and account for post-processing in the definition of secrecy.}

\begin{definition}[Secrecy] \label{def:secrecy}
We say that an agent $A$ has access to \emph{secret operations} $\cT_A^S \subseteq \cT_A$ towards another agent $B$ if
$$  f_B \circ g_A (\rho)\ \sim_B \ f_B (\rho) $$ 
for all $\rho \in \Omega,\ g_A \in \cT_A^S,\ f_B \in \cT_B$. 
If all actions in $\cT_A$ are secret towards $B$ and $\cT_B$ are secret towards $A$ we say that the two agents are \emph{mutually secret}. 
\end{definition}

We may ask if this definition is robust enough, that is, whether further pre- or post-processing by Alice and Bob could destroy the secrecy of a choice of action $g_A\in\cT_A^S$. 
The next proposition shows that no matter how many `secret' transformations in $\cT_A^S$  Alice implements, or how Bob acts in between to try and recover information, he will not detect any of the effects of Alice's actions.
In addition, it is easy to see that pre-processing with a global function (such as distributing entanglement between the two parties) cannot lift secrecy, since  Definition~\ref{def:secrecy} requires it to hold for all initial states.

\begin{restatable}[Robustness of secrecy]{proposition}{PropRobustnessSecrecy}  
\label{prop:secrecy}
If $A$ has secret operations $\cT_A^S$ with respect to $B$ (according to Definition~\ref{def:secrecy}), then pre- and post-processing cannot lift the secrecy, that is
\begin{align*}
 f_B^N \circ \green{g_A^N} \circ \dots  \circ  f_B^2 \circ \green{g_A^2}  \circ f_B^1 \circ \green{g_A^1} \circ f (\rho) \\
\sim_B \
f_B^N \circ \dots \circ  f_B^2  \circ f_B^1
\circ f (\rho), 
\end{align*}
for all states  $\rho \in \Omega$, secret operations $ \{\green{g_A^i}\}_i \subseteq \cT_A^S$ and $ \{f_B^i \}_i \subseteq \cT_B$, global operations $ f \in \cT$ and $N \in \mathbb N$.
\end{restatable}

\subsection{Extended secrecy}

We may also ask whether Alice's actions stay secret to Bob in the presence of an additional global transformation $f\in\cT$. 
Transformations such as a subsystem swap or a communication channel may  break secrecy; others, like the use of a PR box, do not.\footnote{In generalized probability theories, PR boxes can be seen as transformations that take classical inputs and return outputs (Appendix~\ref{appendix:GPT}).}
For this situation, we define an extended notion of secrecy in the spirit of Definition~\ref{def:secrecy}, which reduces to Definition~\ref{def:secrecy} in the case $f = \text{id}$. 
Here, Bob may try to post-process information before and after the global transformation.


\begin{definition}[Extended secrecy]
\label{def:extended_secrecy}
Let $A$ be an agent with access to secret operations  towards  an agent $B$, $\cT_A^S\subseteq \cT_A$ .
We say that $\cT_A^S$ is in addition secret (towards $B$)
\emph{ in the presence of a global transformation} $f\in\cT$ if 
$$ f_B \circ f \circ f'_B \circ g_A (\rho) \ \sim_B \ f_B \circ f \circ f'_B (\rho), $$ 
for all $\rho \in \Omega,\ g_A \in \cT_A^S,\ f_B, f'_B \in \cT_B$.
We say that the agents are \emph{mutually secret in the presence of} $f$ if all actions in $\cT_A$ are secret towards $B$ in the presence of $f$ and vice-versa. 
\end{definition}

We can now show that, analogously to Proposition~\ref{prop:secrecy}, further pre- and post-processing by Alice and Bob cannot lift the secrecy.

\begin{restatable}[Robustness of extended secrecy]{proposition}{PropRobustnessExtendedSecrecy}
\label{prop:extended_secrecy}
If an agent $A$ only uses secret operations $g_A\in\cT_A^S$ with respect to the agent $B$ in the presence of $f\in\cT$, then further pre- and post-processing cannot lift the secrecy, that is
\begin{align*}
    &  \left( \bigcirc_{i=1}^N f_B^i \circ \green{g_A^i} \right)  \circ f \circ  \left( \bigcirc_{i=1}^N {f'}_B^i \circ \green{{g'}_A^i} \right) \circ g(\rho) \\
    &\sim_B \  \left( \bigcirc_{i=1}^N f_B^i \right)  \circ f \circ  \left( \bigcirc_{i=1}^N {f'}_B^i \right) \circ g(\rho)
\end{align*}
for all states  $\rho \in \Omega$, 
local operations $ \{\green{g_A^i}\}_i \subseteq \cT_A$ and $ \{f_B^i \}_i \subseteq \cT_B$, global operations $ g \in \cT$ and $N \in \mathbb N$.
\end{restatable}

In particular, for the case in which Bob only implements post-processing at the very end, Proposition~\ref{prop:extended_secrecy} implies that $\cT_A^S$ forms a monoid.

\begin{corollary}[Secret monoid]
The set $\cT_A^S$ of secret operations in the presence of a global function $f\in\cT$ forms a monoid, i.e.\ $\text{id}\in\cT_A^S$ and
$$ f_A,g_A \in \cT_A^S \implies f_A \circ g_A \in \cT_A^S . $$
\end{corollary}

Naturally, if we further restrict the actions and knowledge of one of the agents (as in Proposition~\ref{prop:nested_agents}), secrecy is maintained.

\begin{restatable}[Restricted agents and secrecy]{corollary}{CorRestrictedAgentsSecrecy}
\label{cor:restricted_agents_secrecy}
Let $A$,  $B$ and $C$ be three agents, such that 
 $C$ is more restricted than $B$, that is $\cT_{C} \subseteq \cT_B$ and  $ [\rho]_{B} \subseteq [\rho]_{C}$, for all $ \rho \in \Omega$.
 
If $\cT_B$ was secret towards $A$ (in the presence of $f \in \cT$), the same is true of $\cT_C$. If $\cT_A$ was secret towards $B$ (in the presence of $f$), it is still secret towards $C$ (idem). 
\end{restatable}

\section{Commuting agents}
\label{sec:commutativity}

\begin{figure*}[t]
    \centering
    {\bf a.}
    \quad
     \includegraphics[width=0.2\textwidth]{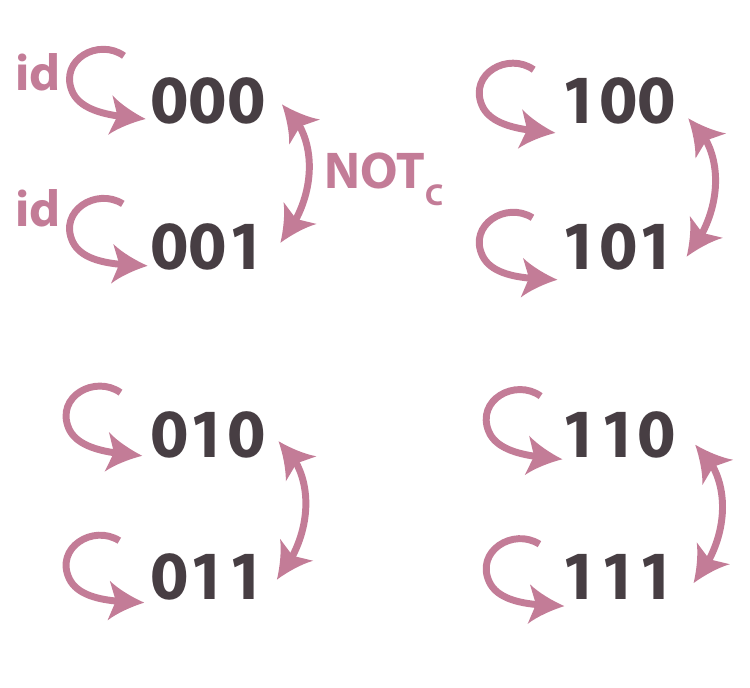}   
     \qquad
     {\bf b.}
    \includegraphics[width=0.2\textwidth]{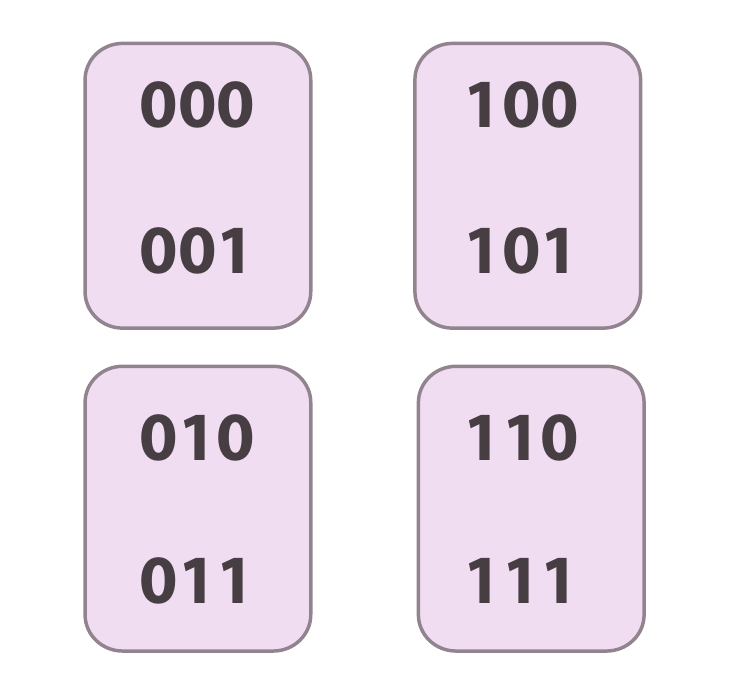}
     \qquad
     {\bf c.}
      \quad
    \includegraphics[width=0.2\textwidth]{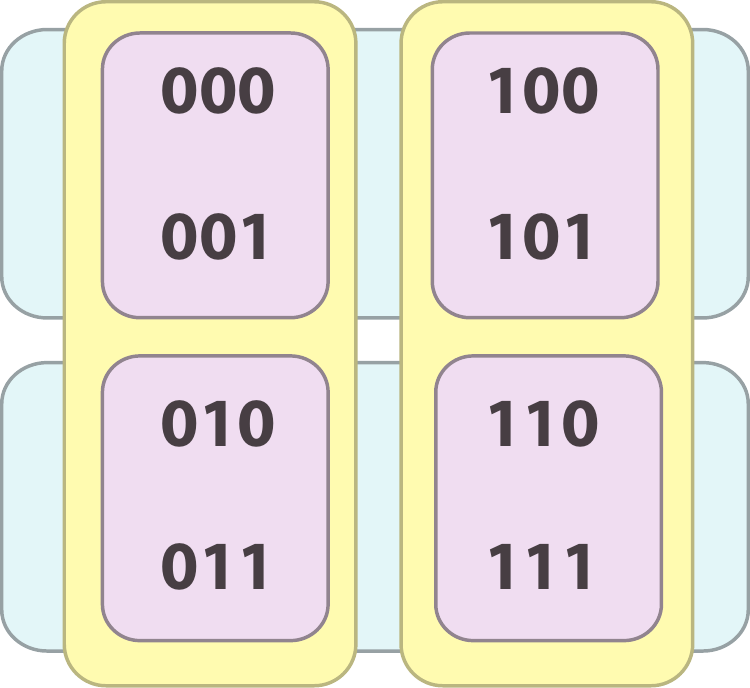}

    \caption{{\bf Three-bit example.} Consider again the theory described by the state space $\Omega$  of 3 bits, and all transformations on those bits. {\bf a.} All  operations $\cT_C$ that change the third bit  (of which $\text{id}$ and $\textsc {not}_C$ are labeled). {\bf b.} Equivalence classes  $[x]_{\cancel C}$ built according to Definition~\ref{def:induced_perspective}. These correspond to the view of an agent who can only distinguish the first two bits. The equivalence relation $\sim_{\cancel C}$, which coarse-grains over the functions applied to the third bit, gives us the largest effective state space relative to which functions in $\cT$ are secret. {\bf c.} More coarse-grained equivalence classes $[x]_A$ (vertical, yellow) and $[x]_B$ (horizontal, blue), corresponding to an agent $A$ who can only distinguish the first bit and an agent $B$ who only sees the second bit, respectively. Operations in $\cT_C$ are still secret relative to these two agents. In addition, operations on the first bit are secret towards $B$ and vice-versa. 
     These smaller effective state spaces correspond to equivalence relations on the effective state space $\Omega_{\cancel C}$ (as in the nested agents of  Proposition~\ref{prop:nested_agents}). The two-bit space $\Omega_{\cancel C}$ is a \emph{common state space} of $A$ and $B$, including states that could be distinguished if the  two agents could work together, with $[x]_{\cancel C} = [x]_A \cap [x]_B$.
}

    \label{fig:nested}
\end{figure*}

Now we explore how secrecy is affected when the actions of two agents $A$ and $B$ commute. This is particularly relevant in the context of the non-signalling principle, since actions at space-like separation naturally commute. 

\begin{definition}[Commuting agents]
\label{def:commuting_agents}
We say that two agents $A$ and $B$ commute if  
$$ f_B \circ g_A (\rho) = g_A \circ f_B (\rho) , $$
for all $  \rho\in\Omega, \ g_A \in \cT_A , \ f_B \in \cT_B$. 
\end{definition}

For example, in field theory commutativity holds for measurements or field interactions at space-like separation, and this is in general how causality is recovered there \cite{Peskin1995}\footnote{The simplest illustration of this is the commutation of the Klein-Gordon field operators 
$\phi(x)$ and $\phi(y)$ at space-like separated $x$ and $y$,
$ [\phi(x),\phi(y)] = 0 .$
Such a commmutation condition is also referred to in field theory as the \emph{locality postulate}~\cite{Banks2008}.}.
Motivated by this, we here take the commutation of actions in space-like separated regions as a fundamental building block in deriving agents that are secret relative to each other. Note that in particular, finding commuting sets of transformations in $\cT$ is something that can be done prior to definitions of local agents; this is shown explicitly in Ref.~\cite{DelRio2015}.\
{Commutation relations result in a nice algebraic structure --- a lattice --- in the space of transformations~\cite{DelRio2015}. This is also the case for the von Neumann bicommutant in operator algebras \cite{Neumann1930}.} Commutation is also an operational property of the theory: for example, commutation is  independent of the choice of  reference frames in relativity and quantum field theory~\cite{Peskin1995,Banks2008}.  
If two agents commute, secrecy follows from simpler conditions.

\begin{restatable}[Secrecy for commuting agents]{proposition}{PropCommutativitySimplifiesSecrecy}
\label{prop:commutation}
If $A$ and $B$ commute, then if there exists a subset of actions $\cT_A^S \subseteq \cT_A$  such that, $\forall \rho\in\Omega,\  g_A \in \cT_A^S, f_B \in \cT_B$, 
   \begin{align*}
   & f_B \circ f\circ g_A (\rho) \ \sim_B \ f_B \circ f (\rho), 
\end{align*}
then $\cT_A^S$ is secret towards $B$ in the presence of $f$. 
In particular, 
$g_A (\rho) \sim_B  \rho$ for all $g_A \in \cT_A, \rho \in \Omega$  implies  secrecy of $A$ towards $B$. 
\end{restatable}

\subsection{Secrecy from commutation}

Starting only from commutation relations on the global transformations, we can construct descriptions of local agents that have secret actions with respect to each other. 
More specifically, given any two commuting submonoids $\cT_A,\cT_B\subseteq\cT$, we can construct equivalence relations $\sim_A,\sim_B$ so that two agents Alice $(\sim_A,\cT_A)$ and Bob $(\sim_B,\cT_B)$ have secret actions with respect to each other. 

The first step is to start with transformations $\cT_A$ (``Alice's transformations"), and look for the most generous effective state space $\Omega_{\cancel A}$ that is insensitive to transformations in $\cT_A$. This will model the perspective of an agent, Bob, who cannot detect Alice's actions. Essentially, this perspective identifies sets of global states that Alice can locally make ``converge'' to the same state.

\begin{definition}[Perspective insensitive to transformations]
\label{def:induced_perspective}
Let $\cT_A \subseteq \cT$  be a submonoid of transformations.  
First we define a binary relation $\sim'_{\cancel A}$\ 
in $\Omega$  called \emph{convergence through $\cT_A$} as 
$$ \rho \sim'_{\cancel A} \sigma \iff \exists \ f_A, g_A \in \cT_A  \text{ s.t. } f_A(\rho) = g_A(\sigma) , $$
We take the transitive closure  $\sim_{\cancel A}$ of  $\sim'_{\cancel A}$  to  define the  \emph{perspective insensitive to transformations} $\cT_A$,  
\begin{align*}
     \rho \sim_{\cancel A} \sigma \iff 
     \exists\ n \in \mathbb N, \ \{\tau_i\}_{i=1}^n  \subseteq \Omega: \\  \rho\sim'_{\cancel A}\tau_1 \sim'_{\cancel A} \tau_2 \sim'_{\cancel A}  \dots \sim'_{\cancel A} \tau_n \sim'_{\cancel A} \sigma.
\end{align*}
\end{definition}

The above construction gives us minimal restrictions for independent agents. 
The following theorem is adapted from~\cite{DelRio2015}. 

\begin{restatable}[Deriving secret agents]{theorem}{ThmDerivingSecretAgents}
\label{thm:secret_agents}
Commuting submonoids $\cT_A,\cT_B\subseteq\cT$ give rise to descriptions of mutually secret agents  $$A=(\sim_{\cancel B},\cT_A), \qquad  B=(\sim_{\cancel A}, \cT_B).$$
\end{restatable}

Indeed, all agents whose actions commute with $\cT_A$ and for whom transformations in $\cT_A$ are secret must be described by a coarse-graining of $\sim_{\cancel A}$ (Figure~\ref{fig:nested}).  This and related minor results can be found in Appendix~\ref{appendix:commuting}. 
In Appendix~\ref{appendix:constructions} we generalize Theorem~\ref{thm:secret_agents} to extended secrecy in the presence of global functions. There, we also extend the construction of the effective spaces of two agents to the case where the two monoids of transformations do not commute: without commutation, this construction is not as simple.

\subsection{Perceived commutation from secrecy}

We can now ask if the actions $\cT_A, \cT_B \subseteq \cT$ of two mutually secret agents must always commute. The answer is no, not at a global level: unbeknownst to the two agents, their actions could affect other degrees of freedom of the global theory. This can become relevant when the actions of two agents affect a common environment that is not directly accessible to them but could be recovered by a third party.

For example, consider again the state-space of three bits, where Alice can only see the first bit and Bob the second. Now imagine that Alice has access to all the transformations that change the first bit and, as a side effect, reset the third bit to $0$, while Bob has access to all the actions that act on the second bit and, as a side effect,  flip the third bit. From a global viewpoint, their actions do not commute. 
However, for someone that only had access to the combined knowledge of Alice and Bob (the first two bits), their actions would appear to commute.
For such an agent, only local time ordering of Alice and Bob's actions matters, as the two processes $f_A \circ f_B$ and $f_B \circ f_A$ are indistinguishable. 
This is yet another example of how subsystems and local descriptions represent simplified pictures of the global theory, reducing the degrees of freedom of the theory to an operational minimum \emph{for a given agent}, who in this case would not need to model global time ordering.

\section{Applications}
\label{sec:discussion}

\begin{figure*}[t]
\centering
\includegraphics{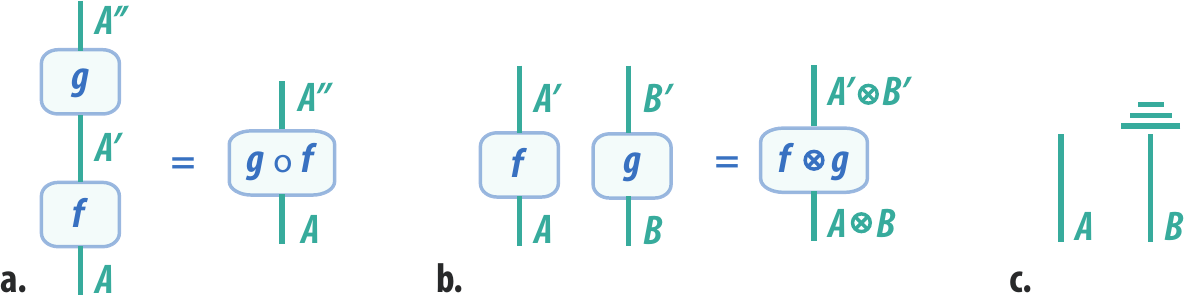}

\caption{  {\bf Process theories.} Processes theories are modular, bottom-up constructions that can be faithfully represented by diagrams~\cite{Coecke2006,Coecke2011,Coecke2014, Coecke2014b, Fritz2015}. Lines represent \emph{systems} (or ``objects'') and boxes \emph{processes} on systems: wires fed in from below a box can be understood as inputs to the process, while wires coming out on top represent the outputs of the process.
Diagrams can be composed due to the strong subsystem structure imposed on process theories, where actions are not seen as affecting the global space but explicitly associated with local systems.
{\bf a.} Processes can be composed in sequence when their output and input systems match. 
{\bf b.} Processes can be composed in parallel on combined systems.
{\bf c.} Discarding a subsystem (e.g.\ taking the partial trace) is indicated by three horizontal lines; in our approach this corresponds to coarse-graining over the relevant degrees of freedom (that is going to a smaller effective space). 
Other conditions can be imposed: e.g.\ in~\cite{Coecke2014b}, causal loops are forbidden, and outputs are always connected to inputs.   
}
\label{fig:process_theories}
\end{figure*}



\begin{figure*}[t]
\centering
\includegraphics{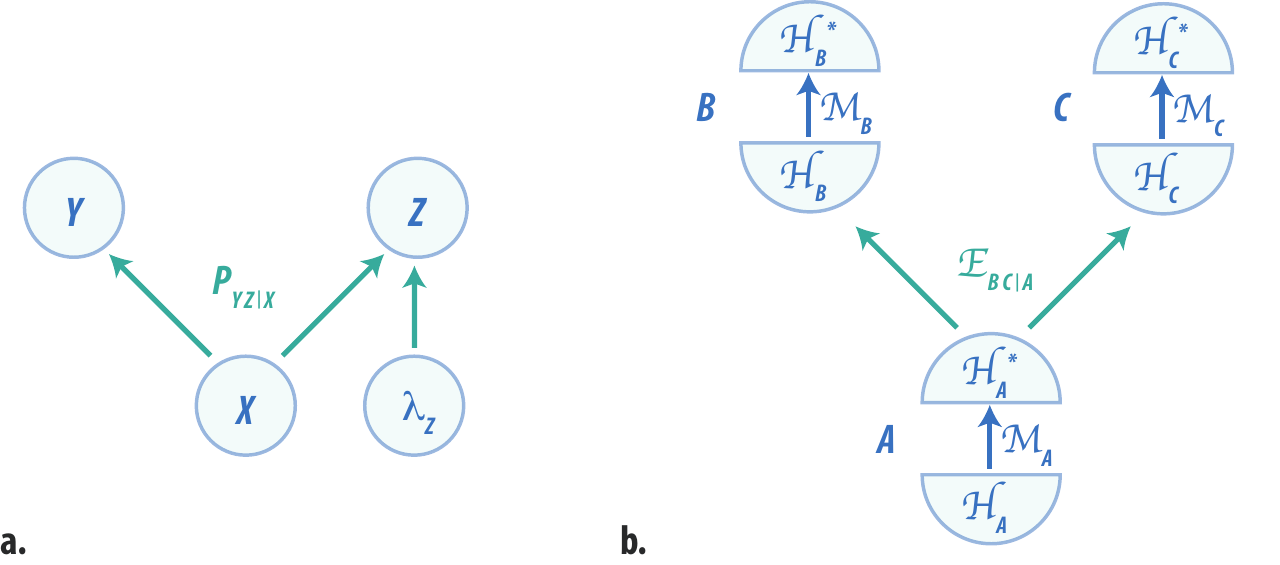}

\caption{  {\bf Causal structures.} 
{\bf a.} In classical causal models, nodes are associated with random variables corresponding to events, while arrows carry causal influence, as specified by conditional probability distributions like $P_{YZ|X}$. These models and distributions may be extended if one later learns of additional causes, like $\lambda_Z$. 
{\bf b.} A model for quantum causal structures proposed in \cite{Allen2016}, where each node is associated with the intervention of an agent in a local space.   A node $i$  is represented by  input and output Hilbert spaces, $H_i$ and $H_i^*$, and by a quantum instrument $\mathcal M_i$ that links the two and corresponds to the agent's intervention (for example, $\mathcal M_i$ could be a local measurement followed by a local preparation dependent on the outcome). 
Causal influence is explicitly carried by quantum maps like $\E_{BC|A}$ which acts like a channel from  $A$ to $B C$. 
}
\label{fig:causality}
\end{figure*}





In the previous sections, we have shown how to derive a notion of locality within a global theory starting from a primitive notion of individual agents, and their observed secrecy and commutation relations.

The operational approach laid out here  has the advantage of carrying very little assumptions about the underlying physical theory. For example, it goes to a higher level of abstraction than generalized probability theories by not taking for granted that all agents express their knowledge in terms of  reliable (classical) statistics about the outcomes of measurements. 

Our notion of effective state spaces  captures the concept of \emph{beables} of a theory: aspects (or classes) of states that can in fact be physically observed and distinguished \cite{Bell2004,Hardy2016}. Our approach highlights that beables are observer-dependent: for example, what appears to be a gauge may turn out to be only a local gauge \cite{Hardy2016}, and the same applies to ``global'' phases of quantum states or yet-to-be-discovered microscopic details of some structure.  We can never rule out the existence of a more refined underlying theory, but with effective state spaces we can 
tailor the descriptions used in a theory to the level of detail needed for a particular application. This goes in the direction of the work of Colbeck and Renner \cite{Colbeck2016}, where it is shown that quantum theory is \emph{complete} for the task of guessing measurement outcomes, and further refinements would be irrelevant.

As presented here, our framework simplifies the modelling of agents for the pedagogical purpose of highlighting the advantages of this general direction. In Appendix~\ref{appendix:relaxing} we show how one could relax some of our assumptions to model agents that are limited in time or who can only approximately distinguish states. 
In the following, we discuss further applications and relation to other work.

\subsection{Non-signalling}

One natural application of our extended notion of secrecy is the traditional non-signalling condition. 
To see this, imagine that the two agents are  cooperating, so that Alice is trying to communicate information to Bob 
by means of some action $g_A\in\cT_A$ on her side. Bob can now either directly apply post-processing $f_B\in\cT_B$, or he can wait for some time to pass, as represented by a function $u_t \in\cT$ that implements global time evolution over time $t$.
If Alice and Bob are mutually secret in the presence of $u_t$, for all $t\leq T$, we conclude that they cannot signal to each other in this time window.

In Appendix~\ref{appendix:GPT} we show explicitly how our notion of extended secrecy implies traditional non-signalling in the framework of generalized probabilistic theories~\cite{Hardy2001, Barrett2007}, where the state space consists of probability distributions over outcomes of possible measurements on physical systems. \

\subsection{Reconstructing space-time}

Building up on the example above, if 
two agents cannot signal in the presence of $u_t$ for $t\leq T$ and 
in addition  can signal in the presence of $u_t$ for $t> T$, this can be used to define a \emph{distance} between the two agents, via $d \propto T$. The proportionality constant can be interpreted as the speed of signal propagation, for example the speed of light. 

The challenge to obtaining a meaningful distance is two-fold: {\bf 1)} choosing a ``natural'' family of transformations $\{u_t \}_t$ to represent time evolutions, and {\bf 2)}  choosing a family of agents that do not conflate different types of coarse-grainings. For example, locality and macroscopicity each give rise to a natural notion of distance, relating to the space between agents and to precision of observation, respectively; the latter could be used to quantify chaos given a family of time evolutions.

More generally, we can try to use signalling between agents to infer properties of space-time of a given theory, as illustrated in Figure~\ref{fig:detectors}.
Some steps in this direction have been given for example in Refs.~\cite{Hardy2016, Hoehn2014, Cao2016}. This would be of particular interest in the context of field theories \cite{Halvorson2006, Wolters2013, Valente2013}. 
We leave the  generalization of the  operational approach depicted in Figure~\ref{fig:detectors} to reconstruct position as future work.

\subsection{Relation to modular approaches}

Our global approach complements modular, bottom-up constructions \cite{Hardy2012},  like process theories based on  symmetric monoidal categories~\cite{Coecke2006,Coecke2011,Coecke2014, Coecke2014b, Fritz2015}. 
For the purpose of comparison with our work, modular theories can be understood as theories of individual systems (or ``objects'') and local actions (``processes'') on those systems, which allow for parallel and sequential composition of processes on different systems. Typically, they assume that:
{\bf 1)}
 Processes with matching output and input systems can be composed sequentially. That is, a process $f: A \to A'$ can be composed with a process $g: A' \to A''$, to form a new process $g \circ f: A \to A''$ satisfying
    $$ g \circ f (A) = g ( f(A)) $$
    (Figure~\ref{fig:process_theories}.a). 
{\bf 2)}
 Any two systems $A$ and $B$ can be combined in parallel to form a composite system denoted by $ A\otimes B $.
 {\bf 3)}
Any two processes $f:  A \to A'$ and $g: B \to B'$  can be composed to yield a process $ f \otimes g: A \otimes B \to A'\otimes B' $  satisfying
$$
(f \otimes g) (A \otimes B) = f(A) \otimes g(B) 
$$ (Figure~\ref{fig:process_theories}.b). 
This last assumption implies that processes act locally without disturbing other systems, and that actions on independent systems always commute. 
This allows us to represent process theories in terms of diagrams that can be easily composed (Figure~\ref{fig:process_theories}). 

Our approach is more general in that we do not assume the strong subsystem structure imposed by conditions {\bf 2)} and {\bf 3)}.
As such, our work strengthens Coecke's argument that non-signalling can be derived from a simpler condition~\cite{Coecke2014b} (Appendix~\ref{appendix:terminality}).  
In general, our top-down view can be taken as a precursor and sanity check for process theories.
In complex global theories, a strong subsystem structure may not be clear cut from the start. The cautious researcher can first use our approach to test different  reduced descriptions for independence conditions. If she succeeds in finding independent effective spaces  --- which is not always possible --- she may then frame them as subsystems and attempt a modular construction.  

At a conceptual level, our approach gives a global interpretation to aspects of process theories that are more epistemic than physical. For example, if we think of subsystems as building blocks of a global space, it appears natural to see ``discarding a subsystem'' as a physical action, like throwing away a piece of lego (Figure~\ref{fig:process_theories}.c). However, if we start from the global space and see subsystems as arbitrary restricted descriptions, then ``discarding a subsystem'' corresponds to a coarse-graining over the relevant degrees of freedom (for example going from $\Omega_{AB}$ to an effective space $\Omega_{A}$), a change of perspective rather than a physical transformation.

\subsection{Relation to causal structures}

Our notion of secrecy between agents is analogous to \emph{causal independence} between events in graphs used to study causality in physics. 
Causal structures~\cite{Penrose1972,Spekkens2012,Pearl2000,Masanes2013,Ried2015,Chiribella2016b} try to capture the causal relations between events 
within a larger context (Figure~\ref{fig:causality}). Both causality (as expressed by Reichenbach's principle) and secrecy are guiding principles of 
a certain way of representing a theory (causal structures and restricted agents respectively) that help us understand a complex situation --- they are 
not necessarily fundamental features of the laws of nature. How useful the representation is depends both on the guiding principle and on the choice of
variables of interest (like events or agents).

Let us illustrate this. 
In classical causal graphs, events are represented by random variables, in principle subject to intervention (Figure~\ref{fig:causality}.a). 
As we move from purely classical scenarios  to more physical situations, like those involving quantum measurements, the formalism of causal structures is evolving 
to focus on agents and on explicit physical transformations as carriers of causal influence, similarly to our approach. 
For example, in the quantum causal structures of~\cite{Allen2016}, events can correspond to quantum systems  where agents can act locally (Figure~\ref{fig:causality}.b). 
Generally speaking, ``events'' embody a particular coarse-graining of a global picture into variables or subsystems of interest. As such, a single causal graph cannot reveal all the features of a complex theory ---  a different decomposition may explore new causal relations.\footnote{Also, we can never know if we only have access to an effective state space, and  there is a deeper theory that changes all the causal relations, e.g.~by providing new common causes.}
The choice of relevant nodes can be guided by \href{https://plato.stanford.edu/entries/operationalism/}{operationalism}: {\bf 1)} we start by picking ``variables'' that we care about (like the outcomes of an experiment, or a subsystem corresponding to the perspective and range of intervention of an agent); {\bf 2)} we then use Reichenbach's principle and independence conditions to complete the causal graph, by identifying further nodes and constraining the channels between them.\footnote{
E.g.~in studying the process of coherent copy $\alpha \ket{0}_A + \beta \ket{1}_A \to \alpha \ket{0}_B \ket{0}_C + \beta \ket{1}_B \ket1_C$  in~\cite{Allen2016}, we start with 3 nodes of interest $A$, $B$ and $C$, and are forced by Reichenbach's principle to complete the graph with a 4th node.}     
This procedure is similar in spirit to how in the present work we could start with the 
description of a few agents and use secrecy and commutation constraints to identify other subspaces and transformations of interest, or build a notion of locality.
How successful we are in this endeavour depends largely on the (subjective) starting point --- a poor initial choice of events or agents could make it impossible to find a meaningful causal graph or independent agents.

Even with a clever choice of initial variables, it could be that the guiding principle is not powerful enough to provide meaningful representations for all physical situations. This is likely the case in both approaches, which are still rooted in classical intuitions --- resulting in concepts like agents, Reichenbach's principle, and time order.  
In trying to explain a physical scenario in terms of these classical notions, we risk running into paradoxes such as the inconsistencies between quantum agents in~\cite{Frauchiger2016}. It remains to explore whether both our approach and causal models can handle this kind of physical challenges, and whether extensions to cover them would still be intuitive enough to help us make sense of the world.

\begin{acknowledgements}
We thank Bob Coecke and Barry Sanders  for discussions on discarding and non-signalling, Lucien Hardy, Ryszard Pawe{\l} Kostecki and 
Renato Renner for discussions on locality, John-Mark Allen and Mirjam Weilenmann for discussions on causality, David Jennings and Markus M\"uller for bringing up approximate distinguishability, Roger Colbeck for feedback on this manuscript, Miyuko for sanctuary, and Sandu Popescu for the mantra ``one idea, one paper''. This one goes out to Matt and Rob --- let's disagree again soon. 

LK is supported by the European Research Council via grant No.\ 258932, the Swiss
National Science Foundation through the National Centre of
Competence in Research \emph{Quantum Science and Technology}
(QSIT),  and   the  European  Commission  via  the  project \emph{RAQUEL}. 
LdR acknowledges support from ERC AdG NLST and EPSRC grant \emph{DIQIP}, from the FQXi grant \emph{Physics of the observer}, and from the
Perimeter Institute for Theoretical Physics. Research at Perimeter Institute is supported by the government of Canada through Industry Canada and by the Province of Ontario through the Ministry of Economic Development \& Innovation.
\end{acknowledgements}


\onecolumngrid

\appendix

\section*{\textsc{Appendix}}

\begin{description}
    \item[Appendix~\ref{appendix:proofs}:] Proofs of all statements in the manuscript.
    \item[Appendix~\ref{appendix:terminality}:] Relation between our secrecy conditions and Coecke's ``non-signalling from terminality'' argument in modular theories \cite{Coecke2014b}.
    \item[Appendix~\ref{appendix:relaxing}:] Relaxing some of the assumptions of the present approach.
    \item[Appendix~\ref{appendix:commuting}:] Additional minor results on commuting agents and minimal constructions.
    \item[Appendix~\ref{appendix:GPT}:] Application to generalized probability theories.
    \item[Appendix~\ref{appendix:constructions}:] Generalizing the derivation of secret agents to extended secrecy in the presence of global functions, and to the case where the two monoids of transformations do not commute.
\end{description}

\vspace{10mm}

\twocolumngrid
\section{Proofs}
\label{appendix:proofs}

\PropNestedAgents*

\begin{proof} 
For each direction:
\begin{itemize}
    \item[$1 \to 2$.] We build the equivalence relation in $\Omega_B$ as 
    \begin{align*}
        [\rho]_B \sim_{B \to C} [\sigma]_B 
        \iff 
        &[\rho']_C =  [\sigma']_C, \\ 
        & \forall\  \rho' \in [\rho]_B, \sigma' \in [\sigma]_B.
    \end{align*} 
    Since  $[\rho]_B \subseteq [\rho]_C$ for all $\rho \in \Omega$, $\sim_{B\to C}$ is a well-defined equivalence relation,  and the reduced space $\Omega_B / \sim_{B \to C}$ is in one-to-one correspondence with the space of the equivalence classes $[\rho]_C$. 
    
    \item[$2 \to 1$.] By assumption, the reduction
    $\h_C$ is isomorphic to $ \h_{B\to C} \circ \h_B. $ 
    Therefore  $[\rho]_C \cong [[\rho]_B]_{B\to C}$, and $[\rho]_B \subseteq [\rho]_C$.
\end{itemize}
\end{proof}

\PropRobustnessSecrecy*

\begin{proof}
We apply the non-signalling condition multiple times.
Define $\rho^{(j)} := \left ( \bigcirc_{i=1}^j f_B^i \circ g_A^i\right)   \circ f (\rho) $.
Starting from the left-hand side, we have 
\begin{align*}
     \left( \bigcirc_{i=1}^N f_B^i \circ g_A^i \right)  \circ f (\rho) 
    &= \ f_B^N \circ g_A^N (\rho^{(N-1)}) \\
    &\sim_B  f_B^N  (\rho^{(N-1)}) \\
    &=\ f_B^N \circ f_B^{N-1} \circ g_A^{N-1} (\rho^{(N-2)}) \\
    & \sim_B f_B^N \circ f_B^{N-1} (\rho^{(N-3)}) \\
    &\vdots \\
    & \sim_B f_B^N \circ \dots \circ f_B^1 \circ f (\rho) .
\end{align*}
\end{proof}

\PropRobustnessExtendedSecrecy*

\begin{proof}
The proof is analogous to the proof of Proposition~\ref{prop:secrecy} and uses Definition~\ref{def:extended_secrecy}; we also define
$$ \tilde \rho := f \circ  \left( \bigcirc_{i=1}^N {f'}_B^i \circ {g'}_A^i \right) \circ g(\rho) $$
and
$$ \rho^{(j)} := \left( \bigcirc_{i=1}^j {f'}_B^i \circ {g'}_A^i \right) \circ g(\rho) . $$
Then
\begin{align*}
   & \left( \bigcirc_{i=1}^N f_B^i \circ g_A^i \right)  \circ f \circ  \left( \bigcirc_{i=1}^N {f'}_B^i \circ {g'}_A^i \right) \circ g(\rho) \\
    &\sim_B f_B^N \circ \dots \circ f_B^1 (\tilde\rho) 
    \quad \flag{\text{secrecy}} \\
   &\sim_B f_B^N \circ \dots \circ f_B^1 \circ f \circ  {f'}_B^N (\rho^{(N-1)})
   \quad  \flag{\text{Def.~\ref{def:extended_secrecy}}} \\
    & \vdots \\
    &\sim_B f_B^N \circ \dots \circ f_B^1 \circ f \circ  {f'}_B^N \circ \dots \circ {f'}_B^1 \circ g(\rho).
\end{align*}
\end{proof}

\CorRestrictedAgentsSecrecy*

\begin{proof}
Since $\cT_C \subseteq \cT_B$, it is secret towards $A$. This also restricts the post-processing that $C$ can do, and since $\rho \sim_B \sigma \implies \rho \sim_C \sigma$, we have that  $\cT_A$ is secret towards $C$.
\end{proof}

\PropCommutativitySimplifiesSecrecy*

\begin{proof}
To show secrecy of $\cT_A$ towards $B$, we have 
\begin{align*}
f_B \circ g_A (\rho) 
&=  g_A \circ \underbrace{f_B (\rho)}_{\in \ \Omega}
 &&\flag{\text{commutativity}} \\
&\sim_B \ f_B (\rho),
\ &&\flag{\text{assumption}} 
\end{align*}
for all $\rho \in \Omega, g_A \in \cT_A, f_B \in \cT_B$.
To show secrecy in the presence of $f$, we use 
\begin{align*}
 f_B \circ f \circ f'_B \circ g_A (\rho)
&=  f_B \circ f \circ g_A \circ \underbrace{f'_B (\rho)}_{\in \ \Omega} \\
&\sim_B  f_B \circ f \circ f'_B (\rho),
\end{align*}
for all $\rho \in \Omega, g_A \in \cT_A, f_B, f'_B \in \cT_B$.
\end{proof}

\begin{restatable}{lemma}{LemmaInducedEquivalence}
The perspective  $\sim_{\cancel A}$ induced by a submonoid $\cT_A \subseteq \cT$ is an equivalence relation in $\Omega$. 
\end{restatable}

\begin{proof}
By construction $\sim_{\cancel A}$ is  transitive,  reflexive and symmetric.
\end{proof}

\ThmDerivingSecretAgents*

\begin{proof}
We must show that 
\begin{align*}
f_B \circ g_A (\rho)\ &\sim_{\cancel A} \ f_B (\rho),
\end{align*}
for all $\rho \in \Omega, g_A\in\cT_A$ and $f_B\in\cT_B$. 
By Proposition~\ref{prop:commutation}, we only need to show 
$ g_A(\rho) \sim_{\cancel A} \rho $,
  for all $\rho \in \Omega,\ g_a \in \cT_A$.
This holds since  $\text{id}\in\cT_A$ (as $\cT_A$ is a monoid), and so  $g_A(\rho)  \sim_{\cancel A} \text{id} (\rho)$.
We proceed analogously to find the effective state space of $A$.
\end{proof}

\section{Relation to terminality}
\label{appendix:terminality}

\begin{figure*}[t]
\centering
\includegraphics{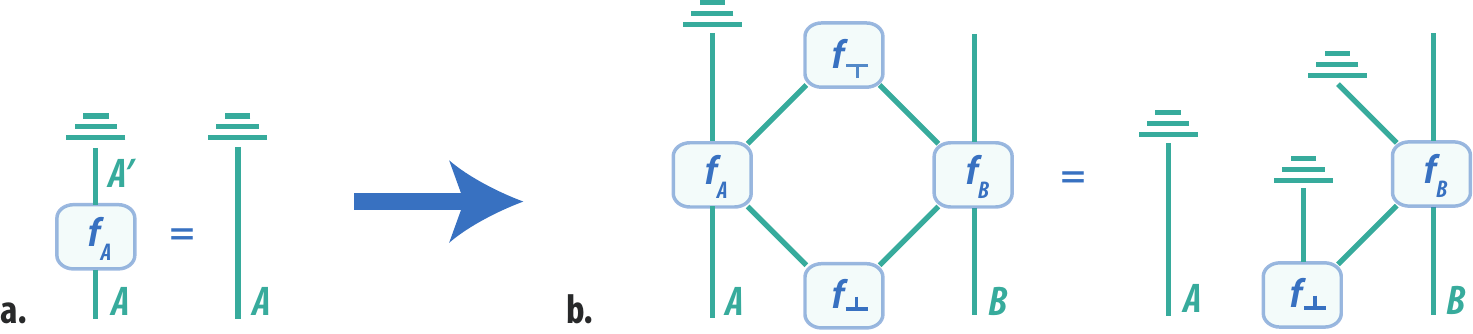}
\caption{ {\bf Terminality and non-signalling in process~\cite{Coecke2014b}.} 
{\bf a.} Terminality: discarding a system $A$ after applying a function $f_A$ is the same as discarding the system directly. 
{\bf b.} Non-signalling: in this setup, a global preparation process $f_\perp$ distributes two systems to Alice and Bob, who perform local operations $f_A$ and $f_B$ respectively on these systems and additional inputs $A$ and $B$. Given terminality, if system $A$ is discarded after this, no information about the original input on $A$ can travel to system $B$. This holds even in the presence of an effect $f_\top$ acting on joint outputs of $f_A$ and $f_B$.}
\label{fig:Coecke}
\end{figure*}

In Ref.~\cite{Coecke2014b}, 
Coecke argues that a process theory \cite{Coecke2006,Coecke2011} is non-signalling if it satisfies a simpler condition dubbed \emph{terminality}. 
Terminality states that local processes on a system right before ``discarding'' it cannot have any observable effect (Figure~\ref{fig:Coecke}.a). Discarding subsystems is a concept that corresponds to tracing out or coarse-graining over local information. For example, in quantum theory it is implemented by the partial trace: terminality is naturally satisfied for  completely positive trace-preserving operations on the discarded systems, but does not hold for non-deterministic effects such as projections onto particular outcomes of measurements~\cite{Coecke2014b}.

In our language, the condition of terminality corresponds to the independence condition $ f_A (\rho) \sim_B \rho, $
where $f_A$ are local functions on a system $A$ and $\sim_B$ corresponds to the local picture of other systems $B$ outside $A$. 
Recall that the assumptions behind process theories like~\cite{Coecke2014b} impose some structure on transformations and agents, in particular commutation between agents' local actions.  As we saw in Proposition~\ref{prop:commutation},  this independence condition together with commutation already implies secrecy.

It remains to see if our secrecy condition is equivalent to the non-signalling of~\cite{Coecke2014b}, depicted in Figure~\ref{fig:Coecke}.b.
This non-signalling corresponds roughly to secrecy under pre- and post-processing, as implemented by an initial  state preparation $f_\perp$ and a deterministic effect $f_\top$. 
In our picture, robustness of secrecy  under pre- and post-processing is ensured by Proposition~\ref{prop:secrecy}. 
In this case, pre-processing with a function $f_\perp$ can be included without loss of generality in the initial state, and post-processing with $f_\top$ is eliminated by the choice of local perspective $\sim_B$. 

Hence, the result of~\cite{Coecke2014b} that terminality implies non-signalling  follows from our propositions \ref{prop:secrecy} and \ref{prop:commutation} together. 
Our premise  that actions by different local agents commute is weaker than the assumptions employed by~\cite{Coecke2014b}.
In conclusion, our approach strengthens the argument in~\cite{Coecke2014b} for the significance of a condition like terminality. 
At the same time, we take a more general approach to subsystems than the bottom-up model of process theories in~\cite{Coecke2014b}, thus highlighting the role of commutation in the context of defining local agents and non-signalling.

\section{Relaxing some assumptions}
\label{appendix:relaxing}

Let us now give some guidelines on how to relax two of the assumptions of our framework, in order to cover more realistic representations of agents.

\subsection{Approximate distinguishability}

Often agents may not have clear-cut distinguishability criteria.
For example, an agent may categorize light frequencies into basic colours such as green and blue --- there may be some frequencies that the agent could file as both green and blue. In the language of PBR~\cite{Pusey2012}, the reduced states ``blue'' and ``green'' would be epistemic and not ontic (with respect to the underlying state space of frequencies $\Omega$).
Agents could also have a notion of approximate distinguishability, for example of the sort ``I can distinguish these two states with probability $1-p$.'' 

We propose a simple approach to address these cases. {\bf 1)} Build a generous effective state space $\Omega_B$ by assigning different reduced states to every two states in $\Omega$ that can be distinguished \emph{in principle} by the agent.   {\bf 2)} Build an  \emph{approximation structures} in the effective state space~\cite{DelRio2015}. An approximation structure comprehends all  neighbourhoods $\{ \ball^\eps(\rho)\}_{\rho \in \Omega_B, \eps \in \E}$ parameterized by whatever measure $\E$ is operational for the  agent. For example, one valid approximation structure for quantum states corresponds to the $\eps$-balls induced by the trace distance; another could be just the \emph{cover} $\{$blue, green, \dots$\}$ of the possible colours assigned to each frequency. {\bf 3)} Build notions of \emph{approximate secrecy}, where we can demand for example
$$  \h_B  \circ  f_B \circ g_A (\rho)\  \in \ball^\eps (\h_B \circ  f_B (\rho)),$$ 
for all $\rho \in \Omega, \ g_A \in \cT_A, \ f_B \in \cT_B$, instead of the stricter condition of secrecy, where we demand that the two final states are completely indistinguishable from $B$'s perspective. 
The properties of approximate secrecy are  inherited from the  approximation structure.

\subsection{Time-limited agents}

In this work we model local actions as monoids $\cT_A$ and $\cT_B$. When applying secrecy to find non-signalling conditions between time-limited agents, the monoidal structure of actions is only a convenient approximation, which allows us to concatenate post-processing actions indefinitely. 
The intuition behind this approximation is that
 Alice and Bob's actions can be implemented essentially instantaneously, compared to the relevant time scales.  One example  would be the action of choosing a bit as an input to a measurement, by pressing a button in Alice's lab (see Appendix~\ref{appendix:GPT}). With this interpretation, $\cT_A$ and $\cT_B$ can consistently be modelled as monoids, because it is assumed that the concatenation of two instantaneous actions can again be implemented instantaneously. 
In this model, time evolution is explicitly modelled by global functions $u_t \in \cT$; this could include the actual effect of pressing the button.

When functions in $\cT_A$ and $\cT_B$ on Alice's and Bob's sides take some finite time $t > 0$ to implement, 
we may instead of full monoidal structure only have 
$ f_A \circ g_A \in \cT $
if the functions $f_A$ and $g_A$ together take less than a given time $T$ to implement. In this case, the notion of secrecy or non-signalling and our results that relate to it can still be recovered for functions and concatenations of functions that do not exceed this time-frame $T$.

\section{Commuting agents: additional results}
\label{appendix:commuting}

In the main text, we have noted that the perspective $\sim_{\cancel A}$ induced by transformations $\cT_A$ is minimal: 
any agent $(\sim_B,\cT_B)$ whose actions $\cT_B$ commute with $\cT_A$ and towards whom transformations in $\cT_A$ are secret must in fact be described 
by a coarse-graining of 
$\sim_{\cancel A}$. 

\begin{proposition}[Induced perspective is minimal]
\label{prop:minimal_perspective}
Let $B=(\sim_B, \cT_B)$ be an agent towards whom $\cT_A$ is secret, and such that $\cT_A$ and $\cT_B$ commute. Then
\begin{align*}
    [\rho]_B &\supseteq [\rho]_{\cancel A}, \quad \forall \ \rho \in \Omega,
\end{align*}
with $\sim_{\cancel A}$ the equivalence relation induced by $\cT_A$. \\ 
This implies that there exists an equivalence relation $\sim_{\cancel A \to B}$ in the effective state space $\Omega / \sim_{\cancel A}$ such that 
$$ \Omega_B \cong (\Omega / \sim_{\cancel A}) / \sim_{\cancel A \to B} . $$
\end{proposition}

\begin{proof}
Since $\cT_A$ is secret towards $B$, 
$$ \rho \sim_B g_A (\rho) $$
and so, due to transitivity of $\sim_B$, also
$$ \exists \ f_A, g_A \in \cT_A \text{ s.t.\ } f_A (\rho) = g_A (\sigma) 
\implies \rho \sim_B \sigma . $$
Again due to transitivity it directly follows that
$$ \rho \sim_{\cancel A} \sigma \implies \rho \sim_B \sigma $$
and so 
$$ [\rho]_B \supseteq [\rho]_{\cancel A}  $$
for all $\rho\in\Omega$. 
We may thus employ Proposition~\ref{prop:nested_agents}.
\end{proof}

\begin{corollary}
Let $\cT_A,\cT_{B} \subseteq \cT$ be monoids such that 
$ \cT_A \subseteq \cT_{B} $. Then the induced equivalence relations
$ \sim_{\cancel A}$ and $ \sim_{\cancel{B}}$ satisfy 
$$ [\rho]_{\cancel{B}} \supseteq [\rho]_{\cancel A}, \ \forall \ \rho \in \Omega . $$
This again implies that there exists an equivalence relation $\sim_{\cancel A \to \cancel B}$ in the effective state space $\Omega / \sim_{\cancel A}$ such that 
$$ \Omega / \sim_{\cancel B} \  \cong (\Omega / \sim_{\cancel A}) / \sim_{\cancel A \to \cancel B} . $$
\end{corollary}
\begin{proof}
This follows from the fact that $\cT_A$ is secret towards $\sim_{\cancel B}$, together with Proposition~\ref{prop:minimal_perspective}. The second statement follows again from Proposition~\ref{prop:nested_agents}.
\end{proof}

Finally, the following proposition shows that equivalence classes $[\rho]_{\cancel A}$ are preserved by commuting transformations $\cT_B$,
providing an operational interpretation to the perspective of agents $(\sim_{\cancel A},\cT_B)$: 
namely, states that are indistinguishable from Bob's point of view remain indistinguishable after he applies functions $f_B \in \cT_B$,
$$ \omega \sim_{\cancel A} \rho \implies f_B (\omega) \sim_{\cancel A} f_B (\rho) . $$

\begin{proposition}[Induced perspective is operational]
\label{prop:local_functions_preserved}
Let $\cT_A,\cT_B \subseteq \cT$ be commuting transformations. Then the perspective $ \sim_{\cancel A} $ induced by $\cT_A$ satisfies
$$ \omega \sim_{\cancel A} \rho \implies f_{B} (\omega) \sim_{\cancel A} f_{B} (\rho) ,$$
for any $f_{B} \in \cT_B$ and $\omega, \rho \in \Omega$.
\end{proposition}

\begin{proof}
From the definition of $ \sim_{\cancel A} $ it follows that
\begin{align*} 
\omega \sim_{\cancel A} \rho \implies &\exists \ n\in \mathbb N, \{\tau_i\}_{1 \leq i\leq n} \text{ with } \tau_i \in \Omega : \\
&\rho \sim'_{\cancel A} \tau_1 \sim'_{\cancel A} \tau_2 \sim'_{\cancel A} \dots \sim_{\cancel A} \tau_n \sim'_{\cancel A} \sigma 
\end{align*}
with
$$ \nu \sim'_{\cancel A} \omega \iff \exists \ f_A,g_A \in \cT_A \text{ s.t.\ } f_A(\nu) = g_A (\omega) . $$
But then, because $\cT_A$ and $\cT_B$  commute,
\begin{align*} 
&\nu \sim'_{\cancel A} \omega \\
\iff 
&\exists \ f_{A}, g_{A} \in \cT_{A} \  \text{ s.t.\ } 
f_{A} ( \nu  ) = g_{A} ( \omega ) \\
\implies 
&\exists \ f_{A}, g_{A} \in \cT_{A} \  \text{ s.t.\ } 
f_{B} \circ  f_{A} ( \nu ) = f_{B} \circ g_{A} ( \omega ) \\
\iff 
&\exists \ f_{A}, g_{A} \in \cT_{A} \  \text{ s.t.\ }
f_{A} \circ f_{B} ( \nu ) = g_{A} \circ f_{B} ( \omega ) \\
\implies 
&f_{B} ( \nu ) \sim'_{\cancel A} f_{B} (\omega ) 
\end{align*}
for all $f_{B} \in \cT_B$. 
From the definition of $\sim_{\cancel A}$ as the transitive closure of $\sim'_{\cancel A}$, it then also follows that
$$ \rho \sim_{\cancel A} \sigma \implies f_B (\rho ) \sim_{\cancel A} f_B (\sigma ) . $$

\end{proof}

\onecolumngrid

\section{Application to GPTs}
\label{appendix:GPT}

Generalized probability theories (GPTs~\cite{Barrett2007,Hardy2001,Wootters1986,Mana2003,Mana2004})  are a framework to infer as much as possible about a physical system without making assumptions about its inner workings (like the assumption that states can be represented as vectors in a Hilbert space). Instead, it is assumed that agents can implement and label a number of physical procedures, like preparations, transformations and, crucially, measurements. Note that the agents  need not know the actual physical state prepared; in order to label a procedure,  they only need to be confident that they can repeat it. 
Indeed, the basic assumption behind GPT frameworks is that agents can extract significant measurement statistics (for example, by repeating a procedure many times). Hence, GPTs model outputs of measurements as random variables, and agents'  knowledge of procedures  as probability distributions.

The usual approach to build GPTs is bottom-up, starting with local procedures that can be composed to reach a global theory.  
Here we are interested in the opposite direction: given a global GPT, can we find meaningful notions of local variables? 
Firstly, we need to model global and local knowledge.

\subsection{Basic formalism}
While we are inspired by known GPT models~\cite{Barrett2007,Hardy2001,Wootters1986,Mana2003,Mana2004}, we take a slightly different and simplified approach here. 
The idea is that agents only have direct access to classical random variables (like input settings and outputs  of a physical measurement). Since they correspond to accessible information, we denote probability distributions over these random variables by \emph{states}. 
\emph{Transformations} $f$ are naturally modelled by conditional probability distributions $P^f_{Z|X}$ that take input to output states, such that 
$f( P_X ) = P'_{Z}$, with   
\begin{align*} 
P'_Z(Z) &= \sum_{x\in X} P^f_{Z|X}(z|x)\  P_X(x).
\end{align*}
For example, suppose that we want to model an experiment where an agent performs a quantum measurement by pressing two buttons: button $X$ prepares a quantum state $\rho^x$  and button $Y$  measures it according to the POVM  $\{ E^y_z \}_z$ with possible outcomes $\{z\}_{z\in Z}$.
The distributions $P_{XY}$ over inputs and $P'_Z$ over outputs correspond to accessible ``states.''  
We model the transformation as a conditional distribution $ P^f_{Z| XY } $ with  $ P^f_{Z|XY }(z|x,y) = \tr (E^y_z\ \rho_x) $. 
The final distribution $P'_Z$ of outcomes given an input distribution $P_{XY}$ is therefore 
\begin{align*}
     P'_Z(z) 
     &= \sum_{x\in X}\sum_{y \in Y} P^f_{Z|XY }(z|x,y)  \ P_{XY}(x,y) \\
     &= \sum_{x\in X} \sum_{y \in Y}  \tr (E^y_z\  \rho_x)  \ P_{XY}(x,y) .
\end{align*}
Note that it is the conditional distribution that encodes the ``physical'' information about a particular setting (like the quantum state and POVM), which may be inaccessible to the agents.  

To compare with the models of~\cite{Hardy2001,Barrett2007}, our 
transformations are analogous to their states. However, in our agent-driven approach, we restrict the set of allowed measurements to those  accessible to a particular agent in the resource theory --- they can thus be seen as a subset of the \emph{fiducial measurements} that define a state in~\cite{Hardy2001,Barrett2007}.

In our model, \emph{global states} correspond to distributions over a global random variable $X$.
Restricted agents are those unable to distinguish some of the outcomes of the global variable. We can model this via arbitrary groupings of outcomes $x \in X$ into equivalence classes, i.e.\ \emph{events} $\{B_b\}_{b}$. 
The \emph{reduction} function $\h_B$ to the effective state space of an agent $B$ simply sums over all the probabilities of the individual outcomes $x\in B_b$ in each event $B_b$ and returns the probability associated with the event, 
\begin{align*}
    P_B = \h_B (P_{X}),\ \text{ with }\ P_B(b) &= \sum_{x\in B_b} P_X(x).
\end{align*}

\subsection{Secrecy and non-signalling}

\begin{figure}[t]
\centering
\includegraphics{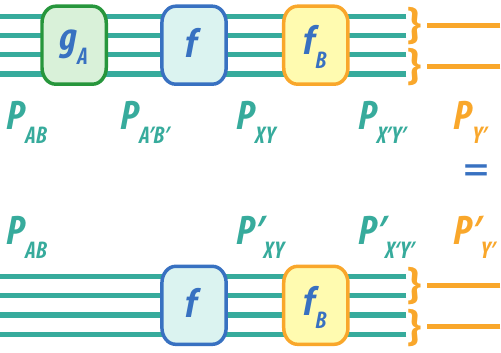}
\caption{ {\bf Extended secrecy in GPTs.}  Consider two agents Alice and Bob and a two-bit random variable $AB$ (the four possible outcomes are represented by vertical lines). Our global state space consists of distributions over two bits (whose names we change after each step to make a proof more readable). 
Let the reduction to the state space of Bob be a coarse-graining over the first bit, $P_Y (y) = P_{XY}(0,y)+P_{XY}(1,y)$, as shown on top. Secrecy in the presence of a function $f$ (Eq.~\ref{eq:extended_gpt}) corresponds to $P_{Y'} = P'_{Y'}$ in the diagram.
This is satisfied, for example if $f$ represents to the use of a PR-box, while $g_A$ and $f_B$ correspond to choices of inputs and post-processing: Bob cannot guess Alice's choice of input after the use of a PR-box. This is equivalent to the traditional notion of non-signalling~~\cite{Colbeck2016} applied to the PR-box. 
}
\label{fig:gpt}
\end{figure}

Consider now two agents $A$ and $B$ whose actions commute. In order to guarantee secrecy of $A$ towards $B$, we only need to satisfy the independence condition $ \green{g_A} (P_X) \sim_B P_X$  (for all global $P_X$ and all $\green{g_A}\in \cT_A$, see Proposition~\ref{prop:commutation}). In our language, this condition reads 
\begin{align*} 
\sum_{y\in B_b} \sum_{x\in X} \green{ P^{g_A}_{Y \vert X}(y \vert x)}\  P_{X}(x) 
&= \sum_{x\in B_b} P_X(x) , \quad \forall\ b.
\end{align*}
For simplicity, we took $Y$ and $X$ to be identical random variables that represent the global state before and after the transformation, and $\h_B$ is a particular coarse-graining of outcomes into events. The condition then states that $\h_B$ is insensitive to the transformation (conditional probability distribution) $P^{g_A}_{Y\vert X}$ from inputs $x\in X$ to outcomes $y\in Y$.

If $A$ and $B$ commute and are mutually secret, we can ask if an additional global transformation $f \in \cT$ allows for signalling between them  (Definition~\ref{def:extended_secrecy}). Our condition for extended secrecy in the presence of $f$,   
$\orange{f_B } \circ \blue{ f} \circ \green{g_A} ( P_X ) \sim_B  \orange{f_B} \circ \blue{ f} ( P_X ),$
for all global $P_X$, $\orange{f_B }\in \cT_B$ and $\green{g_A} \in \cT_A$, becomes
\begin{align}
    \sum_{v\in B_b} \sum_{x,y,z} \orange{ P^{f_B}_{V \vert Z} (v \vert z) }\ \blue{P^f_{Z|Y}(z|y)}\ \green{ P^{g_A}_{Y|X}(y|x)}\ P_X(x)  
    =& \sum_{v\in B_b} \sum_{x,z} \orange{ P^{f_B}_{V\vert Z} (v\vert z) }\  \blue{P^f_{Z|X}(z|x)}\ P_X(x) , \quad \forall \ b.
    \label{eq:extended_gpt}
\end{align}
Non-signalling functions are then those that do not let information encoded in $P^A_{Y|X}$ propagate to Bob's perspective $\sim_B$, that is, such that $\cT_A$ is secret with respect to $(\sim_B,\cT_B)$ in the presence of $f$. Here, again for simplicity $X,Y,Z,V$ were chosen as identical random variables, and $P^f_{Z|X}=P^f_{Z \vert Y}$ both represent the conditional probability distribution corresponding to $f$. For a simple example in a 2-bit space, see Figure~\ref{fig:gpt}.

We can compare our definition of secrecy in the presence of $f$ to traditional notions of non-signalling. Consider again the simple case of a two-bit input and output space of Figure~\ref{fig:gpt}. The definition of non-signalling found for example in~\cite{Colbeck2016} reads $P_{Y|AB}= P_{Y|B}$, or
\begin{align}
    \sum_{x=0,1} P_{XY|AB}(xy|a=0,b) 
    &= \sum_{x=0,1} P_{XY|AB}(xy|a=1,b) , \quad
    \forall \ b, y \in \{0,1\}
   \label{eq:traditional}
\end{align}
This condition formalizes the idea that Bob cannot learn anything about Alice's input $a$ by looking solely at his output $y$ and input $b$. 
In our framework, Alice's choice of input is encoded in a local transformation $g_A \in \cT_A$ (for example, $g_A^0$ could correspond to pressing a button to choose input $0$ and $g_A^1$ to choose $1$), and therefore ``Bob's ignorance about Alice's 
input'' translates to ``Bob's ignorance about Alice's action $g_A$.''

In the following we establish a direct equivalence between these two notions of non-signalling in this simple case; we expect this equivalence to hold in more general settings. Let us first flesh out the assumptions behind the equivalence. A gentle warning: we have labelled all the intermediate bits differently ``to avoid confusion'' (Figure~\ref{fig:gpt}). 
Since we assume a priori that Alice and Bob have  mutual secrecy (without $f$), we take that $g_A$ only acts locally on Alice's bit,  $$P^{g_A}_{A'B'|AB }(a',b'|a,b) = P^{g_A}_{A'|A}(a'|a) \ \delta(b',b), $$
so that 
$$g_A(Q_{AB}(a,b)) = \sum_{a,b} P^{g_A}_{A'|A}(a'|a) \ \delta(b',b)\ Q_{AB}(a,b)=   \sum_a P^{g_A}_{A'|A} (a'|a)\ Q_{AB}(a,b).$$
Similarly, Bob's post-processing is encoded in $f_B \in \cT_B$ which we also assume to be truly local, that is 
$$P^{f_B}_{X'Y'|XY} (x',y'|x,y) = P^{f_B}_{Y'|Y}(y'|y) \ \delta(x',x).$$
The final distribution $P_{Y'}$ for Bob becomes
\begin{align*}
    P_{Y'}(y') 
    &= \orange{\h_B }\circ \orange{ f_B} \circ \blue{ f }\circ \green{g_A} (Q_{AB}(a,b)) \\
    &= \green{\sum_{a}} \ \orange{\h_B}  \circ\orange{ f_B }\circ \blue{ f} ( \green{P^{g_A}_{A'|A} (a'|a)}\ Q_{AB}(a,b)) \\
    &= \green{\sum_{a}} \blue{ \sum_{a',b}}\ \orange{\h_B}  \circ\orange{ f_B } (\blue{P^f_{XY|A'B'}(x,y|a',b) } \  \green{P^{g_A}_{A'|A} (a'|a)}\ Q_{AB}(a,b)) \\
    &= \green{\sum_{a}} \blue{\sum_{a',b}} \orange{ \sum_y }\ \orange{\h_B}  ( \orange{ P^{f_B}_{Y'\vert Y} (y'\vert y)}
    \ \blue{P^f_{XY|A'B'}(x,y|a',b) } \  \green{P^{g_A}_{A'|A} (a'|a)}\ Q_{AB}(a,b)) \\
    &= \green{\sum_{a}} \blue{\sum_{a',b}} \orange{ \sum_y }\orange{ \sum_x }\ \orange{ P^{f_B}_{Y'\vert Y} (y'\vert y)}
    \ \blue{P^f_{XY|A'B'}(x,y|a',b) } \  \green{P^{g_A}_{A'|A} (a'|a)}\ Q_{AB}(a,b).
\end{align*}
The condition for secrecy in the presence of $f$, Eq.~\ref{eq:extended_gpt}, is then 
\begin{align}
    &  \sum_{a,a',b,x,y} 
   \orange{ P^{f_B}_{Y'\vert Y} (y'\vert y)} \ \blue{P^f_{XY|A'B'}(x,y|a',b) }\ \green{P^{g_A}_{A'|A}(a'|a)} \  Q_{AB}(a,b)
   \nonumber \\
    =&  \sum_{a,b,x,y}
    \orange{P^{f_B}_{Y' \vert Y} (y' \vert y)}\ \blue{P^f_{XY|AB}(x,y|a,b)} \ Q_{AB}(a,b),
    \nonumber\\
    &
    \qquad \qquad\qquad \qquad
    \forall \ y'\in \{0,1\}, g_A \in \cT_A, f_B \in \cT_B, Q_{AB}  \in \Omega .
    \label{eq:our_ns_gpt}
\end{align}

\begin{proposition}[Equivalence to non-signalling in GPTs]
In the setting of Figure~\ref{fig:gpt}, ``our'' condition of non-signalling, Eq.~\ref{eq:our_ns_gpt}, is equivalent to the ``traditional'' notion, Eq.~\ref{eq:traditional}.
\end{proposition}

\begin{proof}

In our language, the non-signalling condition of Eq.~\ref{eq:traditional} reads $$P^f_{XY|A'B}(xy|b,a'=0) = P^f_{XY|A'B}(xy|b,a'=1) =: P^f_{XY|A'B}(xy|b) .$$
To show that Eq.~\ref{eq:our_ns_gpt} implies the above, we choose the particular local actions
\begin{align*}
    P^{g^0_A}_{A'|A} &= \delta(a',0),\qquad
    P^{g^1_A}_{A'|A} = \delta(a',1),\qquad
    P^{f_B}_{Y'|Y} (y'\vert y) = \delta(y', y).
\end{align*}
There, $g^0_A$ corresponds to Alice's choice of input $0$, $g^1_A$ to her choice of $1$, and $f_B$ to no post-processing by Bob. 
These choices directly imply for all $Q_{AB}$, 
\begin{align*}
    \sum_{x} P^f_{XY|AB}(xy|b,a=0)\ Q_{B}(b)
    &= \sum_{x} P^f_{XY|AB}(xy|b,a=1)\ Q_{B}(b),
\end{align*}
and so  traditional non-signalling follows.
For the other direction, we have simply
\begin{align*}
    P_{Y'}(y') 
    &= \sum_{a,a',b,x,y\in \{0,1\}}
     \orange{ P^{f_B}_{Y'\vert Y} (y' \vert y)} 
     \  \blue{ P^f_{XY|A'B}(xy|a'b) }\ \green{P^{g_A}_{A'|A}(a'|a) }\ Q_{AB}(ab)  \\
    \flag{\text{non-signalling  (Eq.~\ref{eq:traditional})} } \qquad
    &= \sum_{a,b,x,y} \orange{ P^{f_B}_{Y'\vert Y} (y' \vert y)}  \ \blue{P^f_{XY|B}(xy|b) }
     \ \underbrace{\left[  \green{P^A_{A'|A}(0|a)}
    + \green{P^A_{A'|A}(1|a)}\right]}_{=1}
    \  Q_{AB}(a,b) \\ 
    &= \sum_{a,b,x,y} \orange{P^B_{Y'\vert Y} (y' \vert y) }
    \  \blue{P^f_{XY|AB}(xy|b) } \ Q_{AB}(a,b) .
\end{align*}

\end{proof}

This shows that for example PR boxes satisfy our definition of non-signalling functions. Examples for functions that are signalling are conditional probability distributions that swap the states on the two systems, or bitwise addition of the inputs $a$ and $b$ on the two sides into the outputs $x$ and $y$.

\section{Deriving secrecy without commutation}
\label{appendix:constructions}

In principle, the way we have constructed an induced perspective $\sim_{\cancel A}$ from a monoid $\cT_A$ can be extended to construct equivalence relations that yield secret agents 
$(\sim_{A},\cT_A)$ and $(\sim_{B}, \cT_B)$ 
in the presence of a global function $f$, and even in the case where functions $\cT_A$ and $\cT_B$ do not commute. 
This is done in the following proposition, which, as is shown in the subsequent corollary, reduces to the definition of induced perspectives $\sim_{\cancel A}$ and $\sim_{\cancel B}$ in the case $f = \text{id} $ and commuting $\cT_A,\cT_B$.

\begin{proposition}[Deriving secret agents]
\label{prop:deriving_secret_agents}
Let $(\Omega,\cT)$ be a gobal theory, $\cT_A^S,\cT_B \subseteq \cT$ be two monoids of transformations, 
and let $f \in \cT$. \\
Then the smallest equivalence class $\sim_B$ on $\Omega$ towards which $\cT_A^S$ is secret in the presence of $f$,
$$ f_B \circ f \circ f'_B \circ g_A (V) \sim_B  f_B \circ f \circ f'_B (V), $$ 
for all $V \in S^\Omega, g_A \in \cT_A^{S}, f_B \in \cT_B$, 
is built as follows: \\
Define the relation $\sim$ on $\Omega$ as
$$ \rho \sim \sigma \iff 
\exists \ f_A, g_A \in \cT_A^S \text{ s.t.\ } 
f_A( \rho ) = g_A ( \sigma ) . $$
Then define another relation $\sim'$ as
\begin{align*}
\rho \sim' \sigma \iff 
\begin{cases} 
\rho \sim \sigma \text{ or } \\
\exists \ \rho', \sigma' \in \Omega, f_B \in \cT_B \text{ s.t.\ } 
\rho  = f_B ( \rho' ),  \sigma = f_B ( \sigma' ),  \rho' \sim \sigma'  \text{ or } \\
\exists \ \rho', \sigma' \in \Omega, f_B, f'_B \in \cT_B \text{ s.t.\ } \rho  = f_B \circ f \circ f'_B ( \rho'  ),  \sigma = f_B \circ f \circ f'_B ( \sigma' ),  \rho' \sim \sigma'.
\end{cases} 
\end{align*}
Finally, the relation $\sim_B$ on $\Omega$ is the transitive closure of $\sim'$, namely through
$$ \rho \sim_B \sigma \iff \exists \ n \in \mathbb N, \tau_1, \dots, \tau_n \in \Omega 
\text{ s.t.\ } \rho \sim' \tau_1, \tau_1 \sim' \tau_2, \dots, \tau_n \sim' \sigma . $$
\end{proposition}
\begin{proof}
Both the relation $\sim$ and $\sim'$ are by construction reflexive and symmetric. The relation $\sim_B$ is then by construction also transitive, and thus constitutes an equivalence relation. The relation $\sim_B$ furthermore gives rise to the smallest perspective towards which $\cT_B$ is secret in the presence of $f$: 
$ \omega = g_A ( \rho ) \implies \rho \sim_B \omega .$
By construction then also
$ f_B \circ f \circ f'_B \circ g_A ( \omega ) \sim_B  f_B \circ f \circ f'_B ( \omega ) $
and 
$ f_B \circ g_A ( \omega ) \sim_B f_B ( \omega ) $,
for all $\omega \in \Omega, f_B,f'_B \in \cT_B, g_A \in \cT_A^S$.
\end{proof}

\begin{corollary}
In the case of commuting $\cT_A^S,\cT_B\subseteq \cT$, the equivalence relation $\sim_B$ constructed in Proposition~\ref{prop:deriving_secret_agents} that gives rise to secrecy in the presence of $f$ simplifies accordingly and can be constructed as follows. \\
Define the relation $\sim$ on $\Omega$ as
$$ \rho \sim \sigma \iff \exists \ f_A, g_A \in \cT_A^S \text{ s.t.\ } 
f_A ( \rho ) = g_A ( \sigma ) . $$
Then define another relation $\sim'$ as
\begin{align*}
\rho \sim' \sigma \iff 
\begin{cases} 
\rho\sim \sigma \text{ or } \\
\exists \ \rho', \sigma' \in \Omega, f_B \in \cT_B \text{ s.t.\ } 
\rho = f_B \circ f ( \rho' ),  \sigma = f_B \circ f ( \sigma' ),  \rho' \sim \sigma' .
\end{cases} 
\end{align*}
Then, the relation $\sim_B$ on $\Omega$ is the transitive closure of $\sim'$, namely through
$$ \rho \sim_B \sigma \iff 
\exists \ n \in \mathbb N, \tau_1, \dots, \tau_n \in \Omega 
\text{ s.t.\ } 
\rho \sim' \tau_1, \tau_1 \sim' \tau_2, \dots, \tau_n \sim' \sigma . $$
If in addition $f=\mathbb I$, the relation $\sim_B$ simplifies to
$$ \rho \sim_B \sigma \iff \exists \ n \in \mathbb N, \tau_1, \dots, \tau_n \in \Omega 
\text{ s.t.\ } 
\rho \sim  \tau_1, \tau_1 \sim \tau_2, \dots, \tau_n \sim \sigma $$
with $\sim$ as above. This recovers the construction of induced perspectives $\sim_{\cancel A}$ in Definition~\ref{def:induced_perspective}.
\end{corollary}

\begin{proof}
In the case when functions in $\cT_A^S$ and $\cT_B$ commute, we can see that
\begin{align*} 
\rho \sim \sigma &\iff   \exists \ f_A, g_A \in \cT_A^S \text{ s.t.\ } 
f_A ( \rho ) = g_A ( \sigma ) \\
&\implies 
\exists \ f_A, g_A \in \cT_A^S \text{ s.t.\ } f_B \circ f_A( \rho ) = f_B \circ g_A ( \sigma ) \\
&\iff \exists \ f_A, g_A \in \cT_A^S \text{ s.t.\ }
f_A \circ f_B ( \rho ) =  g_A \circ f_B ( \sigma ) \\
&\iff f_B ( \rho ) \sim f_B ( \sigma ) 
\end{align*}
for all $f_B \in \cT_B$. This implies the respective simplifications of the relation $\sim_B$, and recovers the induced perspective $\sim_{\cancel A}$ of $\cT_A^S$ in the case of $f= \text{id} $.
\end{proof}


\vspace{10mm}

\end{document}